\documentclass{article}
\usepackage[utf8]{inputenc}
\usepackage{proceed2e}
\usepackage[numbers]{natbib}
\usepackage{amsmath,amssymb,amsfonts,amsthm}
\usepackage{graphicx}
\usepackage{textcomp}
\usepackage{algorithmic}
\usepackage{algorithm}
\usepackage{array}
\usepackage{hyperref}
\PassOptionsToPackage{hyphens}{url}\usepackage{hyperref}

\newtheorem{theorem}{Theorem}
\theoremstyle{definition}
\newtheorem{defn}{Definition}[section]

\newtheorem{lemma}{Lemma}[section]

\newcommand{\M}{\mathcal{M}}

\def\figfolder {.}
\DeclareGraphicsExtensions{.pdf}

\title{Differentially private cross-silo federated learning}

\author{\textbf{Mikko Heikkil\"a}$^1$, \textbf{Antti Koskela}$^2$, \textbf{Kana Shimizu}$^3$, \textbf{Samuel Kaski}$^{4,5}$, and \textbf{Antti Honkela}$^{2}$\\
  {$^1$ Helsinki Institute for Information Technology HIIT, Department of Mathematics and Statistics, University of Helsinki, Finland} \\
  {$^2$ Helsinki Institute for Information Technology HIIT, Department of Computer Science, University of Helsinki, Finland} \\
  {$^3$ Department of Computer Science and Engineering, Waseda University, Japan} \\
  {$^4$ Helsinki Institute for Information Technology HIIT, Department of Computer Science, Aalto University, Finland} \\
  {$^5$ Department of Computer Science, University of Manchester, UK}
  }

\begin{document}

\maketitle

\begin{abstract}
Strict privacy is of paramount importance in distributed machine learning. 
Federated learning, with the main idea of communicating only what is needed for learning, 
has been recently introduced as a general approach 
for distributed learning to enhance learning and improve security. 
However, federated learning by itself does not guarantee 
any privacy for data subjects. To quantify and control 
how much privacy is compromised in the worst-case, we can 
use differential privacy.

In this paper we combine additively homomorphic secure summation protocols 
with differential privacy in the so-called cross-silo federated learning setting. The goal is to  
learn complex models like neural networks while guaranteeing strict privacy 
for the individual data subjects. We demonstrate that our proposed solutions give 
prediction accuracy that is comparable to the 
non-distributed setting, and are fast enough to enable learning models with millions of parameters in a reasonable 
time.

To enable learning under strict privacy guarantees that need privacy amplification by subsampling, 
we present a general algorithm for oblivious distributed subsampling. 
However, we also argue that when malicious parties are present, 
a simple approach using distributed Poisson subsampling gives better privacy.

Finally, we show that by leveraging random projections we can further scale-up our approach 
to larger models while suffering only a modest performance loss.
\end{abstract}

\section{Introduction}
\label{sec:introduction}

Privacy is increasingly important for modern machine learning.
Commonly, the sensitive data needed for learning are distributed to different parties, such as 
individual hospitals and companies or IoT device owners. Federated learning (FL) \cite{McMahan_2017} 
has been recently introduced in order to reduce the risks to privacy and improve model training. The 
key idea in FL is that only the information necessary for learning a model will be communicated while 
the actual data stays distributed on the individual devices.

FL makes attacking the system harder since no centralised server holds all the data. 
Nevertheless, it has been recently demonstrated that FL by itself is not 
enough to guarantee any level of privacy \cite{zhu2019deep}. The current gold standard in privacy-preserving 
machine learning is differential privacy (DP) \cite{dwork_et_al_2006}, 
which is based on clear mathematical definitions and aims to guarantee a level of indistinguishability for 
any individual data subject: the results of learning would be nearly the same if any single data entry 
was removed or arbitrarily replaced with another.

In this paper, we concentrate on FL with DP to enable efficient learning on distributed data 
while guaranteeing strong privacy. We focus on the so-called cross-silo FL \cite{Kairouz_2019}, 
where the number of parties might range from a few to tens of thousands, and 
each party generally holds from tens to some thousands of samples of data. 
Additionally, we consider using trusted execution environments (TEEs) as 
an extra layer of security.

To guarantee strict privacy with good utility, it is often vital that privacy can be amplified 
by the stochasticity resulting from using data subsampling \cite{Kasiviswanathan_2011}. This requires that  
the origin of the samples in a minibatch remains hidden from other parties. We therefore introduce an  
algorithm that enables oblivious distributed subsampling with an arbitrary sampling scheme. 

The two most used sampling schemes with privacy amplification are 
sampling without replacement, where the minibatch size 
is fixed and each sample has equal probability of being included, and Poisson sampling, 
where the probability of independently including each sample is fixed 
and the batch size varies. 
We show that when malicious parties are present, using a simple distributed Poisson sampling scheme 
generally gives better privacy than using a more complex distributed 
sampling without replacement scheme.

Finally, we argue both theoretically and on empirical grounds 
that by using random projections, we can scale our methods up to even larger models while only 
incurring a modest loss in prediction accuracy.

The code for all the experiments is freely available \cite{mikko_heikkila_2020_3938639}.


\subsection{Contribution}

We consider DP cross-silo FL in two settings: 
with pairwise connections between the parties (feasible for at most tens to hundreds of parties), 
and with connections to central servers (hundreds to tens of thousands of parties).

We present and empirically test two methods based on combining fast secure multiparty-computation (SMC) 
protocols and DP, with emphasis on iterative learning and distributed subsampling. The 
methods can be used for learning complex models such as neural networks from large datasets under 
strict privacy guarantees, with reasonable computation and communication capacities. 
The proposed method without pairwise connections 
is based on a novel secure summation algorithm, 
preliminarily presented in an earlier conference publication \cite{Heikkila_2017}. 

For achieving privacy amplification by subsampling in the distributed setting, we propose an 
algorithm for oblivious subsampling under an arbitrary sampling scheme. 
We also show that 
a simpler distributed Poisson sampling scheme is preferable to a more complex scheme of 
distributed sampling without replacement, when malicious parties are present.

To speed up encryption and reduce the amount of communication needed, we also consider 
communicating only randomly projected low-dimensional representations. We argue theoretically and 
show with experiments that the effect on utility under DP is small.


\section{Background}

In the next sections we give a short review of the necessary theory, starting with DP in 
Section \ref{sec:DP}, and continuing with random projections in Section \ref{sec:dim_reduction}, 
and trusted execution environments in Section \ref{sec:TEEs}.

\subsection{Differential privacy}
\label{sec:DP}

Differential privacy \cite{dwork_et_al_2006, dwork_roth_2014} has emerged as the leading method 
for privacy-preserving machine learning. On an intuitive level, DP guarantees a level of deniability to each data subject: 
the results would be nearly the same if any single individual's data were replaced with an arbitrary sample. More 
formally we have the following:

\begin{defn} (Differential privacy)
\label{def:ed-dp}
A randomised mechanism $\M: \mathcal{D} \rightarrow \mathcal{R}$ is $(\epsilon,\delta)$-DP, if for all 
neighbouring datasets $D, D' \in \mathcal{D}: D \sim D'$ and 
for any measurable $E \subseteq \mathcal{R}$, 
$$ P(\M (D) \in E ) \leq e^{\epsilon} P(\M (D') \in E ) + \delta, $$
where $\epsilon > 0, \delta \in [0,1]$. If $\delta=0$, the mechanism is called (pure) $\epsilon$-DP.
\end{defn}

We consider two specific neighbouring relations: when $| D |~=~|D'|$ and they differ by a single element (so-called bounded DP), 
the relation is called a substitution relation and denoted by $\sim_S$. 
When one can be transformed into the other by removing or adding a single element (unbounded DP), 
the relation is called a remove/add relation and is denoted by $\sim_R$.

One very general method for guaranteeing DP is the Gaussian mechanism, which amounts to 
adding independent Gaussian noise to the learning algorithm. The amount of privacy then 
depends on the sensitivity of the function and the noise magnitude:

\begin{defn} (Sensitivity)
\label{def:sensitivity}
Let $f: \mathcal{D} \rightarrow \mathbb R^d$. The $\ell_2$-sensitivity of $f$ is defined as
\begin{equation}
\Delta = \sup_{D,D' \in \mathcal D:D \sim D' } \| f(D) - f(D')  \|_2 ,
\end{equation}
where $\sim$ denotes a general neighbourhood relation.
\end{defn}

\begin{defn} (Gaussian mechanism)
\label{def:Gaussian_mechanism}
Let $f: \mathcal{D} \rightarrow \mathbb R^d$ with sensitivity $\Delta$. A randomised mechanism 
$\mathcal G_f: \mathcal D \rightarrow \mathbb R^d$, 
\begin{equation}
\mathcal G_f (D) = f(D) + \mathcal N (0, \sigma^2 \cdot I_d)
\end{equation}
is called the Gaussian mechanism.
\end{defn}

Composition refers to DP mechanism(s) accessing a given dataset several times, which 
either weakens the privacy guarantees or requires increasing the noise magnitude guaranteeing privacy.
No analytical formula is known for the exact privacy guarantees provided by 
the Gaussian mechanism for a given number of compositions when using data subsampling. 
However, as shown in \cite{Sommer2019, Koskela2019}, they 
can be calculated numerically with arbitrary precision.

In effect, for some function $f$, we can numerically calculate the amount of privacy for a given 
number of compositions when using the subsampled Gaussian mechanism with 
a chosen sequence of noise values. This is called privacy accounting. In practice, we use the 
privacy loss accountant of \cite{Koskela2019} for all the actual privacy 
calculations\footnote{https://github.com/DPBayes/PLD-Accountant}.


\subsection{Dimensionality reduction}
\label{sec:dim_reduction}

Using SMC protocols is generally expensive 
in terms of communication and computation, and even with faster SMC protocols the 
overhead can be significant. We therefore also consider dimensionality reduction by a random projection 
to scale the methods to even larger problems.

Instead of using SMC protocols 
to send a vector of dimension $d$, 
each party only communicates a projection with dimension $k$ s.t. $k \ll d$. 
The representation is obtained using a random projection matrix generated with a shared seed.

For the projection, we use a mapping $f_{JL}(u) = P^T u, P \in \mathbb R^{d \times k}$ s.t. 
each element of $P$ is drawn independently from $\mathcal N(0,1/k)$. 
Given the low-dimensional representation $P^T u$ and $P$, since $\mathbb E_P [P P^Tu] = u$, 
the original vector $u$ can be approximated as $u \simeq P P^Tu$. 
Moreover, as shown in \cite{Dasgupta_2003}, with large enough $k$ the 
Johnson-Lindenstrauss lemma \cite{JL_1984} guarantees that $u^T P P^T u \simeq u^T u$. 
These properties explain our experimental results in 
Section \ref{sec:experiments} showing that when training under DP 
and with $k \ll d$, learning with projections results in nearly the same prediction accuracy 
as using full-dimensional gradients.

To guarantee privacy when the privacy mechanism operates on a low-dimensional 
projection and the projection matrix is public, we first define sensitivity bounded 
functions and then state the essential lemma, both introduced in \cite{Agarwal_2018}:
\begin{defn}[Function sensitivity]
\label{def:function_sensitivity}
A randomised function $f:\mathcal D \rightarrow \mathcal X $ is $(\Delta_f, \delta)$-sensitive, if 
for any $D,D' \in \mathcal D:D \sim D'$, there exist coupled random variables 
$X, X' \in \mathcal X$ s.t. the marginal distributions of $X, X'$ are identical to those of 
$f(D), f(D')$, and 
$$ \mathbb P_{X, X'} ( \| X - X' \|_2 \leq \Delta_f ) \geq 1-\delta .$$
\end{defn}

\begin{lemma}
\label{lemma:dp_projection}
Let $\mathcal M: \mathcal X \rightarrow \mathcal R$ be an $(\epsilon, \delta)$-DP mechanism for 
sensitivity $\Delta_f$ queries, and let $f:\mathcal D \rightarrow \mathcal X$ be a $(\Delta_f,\delta')$-sensitive function. 
Then the composed mechanism $\mathcal M ( f(D) )$ is $(\epsilon, \delta + \delta')$-DP.
\end{lemma}


 \subsection{Trusted execution environments}
 \label{sec:TEEs}

When the parties running the distributed learning protocols 
have enough computational resources available, we can increase 
security by leveraging trusted execution environments (TEEs).

TEEs aim to provide a solution to the general problem of 
establishing a trusted computational environment 
by providing confidentiality, integrity, 
and attestation for any computations run inside the TEE: no adversary outside the environment 
should 1) be able to gain any information about the computations done inside it (confidentiality), 
nor 2) be able to influence the computations done by the TEE 
(integrity), and 3) the TEE should be able to give an irrefutable and unforgeable proof 
that the computation has been done inside it
(attestation) \cite{Subramanyan_2017}.

In effect, TEE provides a securely isolated area in which users can store sensitive data and run the code they want protected. TEE is realised with a support of hardware, and it aims to protect applications from software and hardware attacks. 
Currently, technologies such as Intel's Software Guard Extensions (SGX) and ARM's TrustZone are used for implementing TEE.

It should be emphasised that the security of TEEs depends not only on software but also on the hardware manufactures, 
and should not be generally accepted at face value; various attacks targeting TEEs have been 
reported, including arbitrary code executions ~\cite{ rosenberg2014qsee, shen2015exploiting } and 
side-channel attacks~\cite{xu2015controlled, van2018foreshadow, cho2018prime+ }.

In this paper we assume TEEs are used by well-resourced parties as detailed in 
Section \ref{sec:cross_silo_dp_learning}. However,  we do not rely exclusively on TEEs for 
doing secure  computations, but only propose them as an additional layer of security to make 
attacking the system harder.


\section{DP cross-silo federated learning}
\label{sec:cross_silo_dp_learning}

In the general setting we consider, there are $N$ data holders (clients), the 
data $x$ are horizontally partitioned, i.e., each client 
has the same features, and the $i$th client has $n_i$ data points. 
The privacy guarantees are given on an individual sample level.

Our main focus is learning complex models using gradient-based optimisation. 
Using the standard empirical risk minimisation framework \cite{Vapnik_1991_ERM}, 
given a loss function $L(h(x)), h \in \mathcal H$, where $\mathcal H$ is some function class and 
$x$ includes the target variable, 
we would like to find $h^*_{OPT} = \arg\min_{h \in \mathcal H} L(h(x))$. However, since we do not know the 
underlying true data distribution, we instead minimise the average loss over data 
$h_{OPT} = \arg\min_{h \in \mathcal H} 1/n \sum_{i=1}^N \sum_{j=1}^{n_i} L(h(x_{ij}))$, where $n=\sum_{i=1}^N n_i$ 
and $x_{ij}$ refers to the $j$th sample  from client $i$.

The standard method for solving this problem with privacy in the centralised setting is 
DP stochastic gradient descent (DP-SGD) 
\cite{Song_2013, Abadi2016}. In this case, on each iteration we need to calculate 
\begin{equation}
\label{eq:dp_sgd_noisy_sum}
1/b [ \sum_{i=1}^N \sum_{j \in [b_i]} \triangledown_{\theta} \tilde L(x_{ij} | \theta) + \eta ], 
\end{equation}
where $[b]=\cup_{i} [b_i]$ is the chosen batch of size $b$, 
$\triangledown_{\theta} \tilde L$ is the clipped gradient of the loss, 
$\theta$ is the model parameter vector, 
and $\eta$ is the noise that guarantees privacy (see Definition \ref{def:Gaussian_mechanism}). 
In the following, we write $z_{i j} \stackrel{\Delta}{=} \triangledown_{\theta} \tilde L(x_{ij} | \theta)$.

With small problems, \eqref{eq:dp_sgd_noisy_sum} can be readily computed 
with general secure multiparty computation protocols, as done e.g. in \cite{Jayaraman_2018}. 
In our case this is not feasible due to high computational and communication requirements.

From the privacy perspective, an ideal distributed learning algorithm would have DP noise magnitude 
equal to the corresponding non-distributed learning algorithm 
run by a single trusted party on the entire combined dataset. We refer 
to this baseline as the trusted aggregator version.

Another important property for distributed private learning, that is mostly orthogonal to the amount of noise added, 
is graceful degradation: when the privacy guarantees are affected by adversarial parties, 
we want to avoid catastrophic failures where a single party can (almost) completely 
break the privacy of others. Instead, the privacy guarantees should degrade gracefully with 
the number of adversaries.

It is readily apparent from \eqref{eq:dp_sgd_noisy_sum} that the problem can be broken into two largely 
independent components: distributed noise addition and secure summation. We will consider these 
problems separately in sections \ref{sec:distributed_noise} and \ref{sec:secure_sum}, respectively.


\subsection{Threat model}
\label{sec:threat_model}

We allow for three types types of parties: honest, (non-colluding) honest-but-curious (hbc) and malicious.

The only honest parties we include are TEEs. We also consider the 
possibility of an adversary gaining full control of some number of TEEs, 
which effectively turns these into malicious parties.

The privacy guarantees of the methods are valid against fully malicious parties, whereas 
the accuracy of the results is only guaranteed with malicious parties who avoid 
model poisoning attacks or disrupting the learning altogether.


\subsection{Distributed noise addition}
\label{sec:distributed_noise}

In order to guarantee privacy, we need to add noise 
to the sum as in \eqref{eq:dp_sgd_noisy_sum}. This can be done utilising infinite 
divisibility of normally distributed random variables: when client $i$ 
adds Gaussian noise with variance $\sigma^2_i$ the total noise after summation is Gaussian with variance 
$\sum_{i=1}^N \sigma^2_i$.

Given a target noise variance $\sigma^2$, using TEEs 
we can simply set $\sigma_i^2 = \sigma^2/N$ so 
the sum corresponds to the trusted aggregator noise level. 
However, in case a TEE is compromised, 
the attacker can remove $1/N$ fraction of the noise from the results. 
This weakens the privacy guarantees, but does not typically result in a sudden catastrophic 
failure.

For the case without TEEs (or to protect against compromised TEEs with obvious modifications), 
we can upscale the noise to introduce a trade-off 
between the noise level and protection against colluders: 
assuming $T$ colluders, we need 
$$\sigma^2_i \geq \frac{1}{N-T-1} \sigma^2$$ 
(see e.g. \cite{Heikkila_2017} for a proof).

The total noise level in this case is sub-optimal compared to the trusted aggregator setting, but 
for moderate $T$ and larger $N$ the noise level is still close to optimal. 
Even if the actual number of colluders exceeds the parameter $T$ used in the algorithm, 
the privacy guarantees still weaken gracefully with the number of malicious parties.


\subsection{Secure summation}
\label{sec:secure_sum}

We consider two fast secure summation protocols. If the number of clients is low (at most some dozens), 
we can use fast homomorphic encryption based on secret sharing using pairwise keys. 
The algorithm was originally introduced in \cite{Castelluccia_2005}, uses fixed-point representation of real numbers, and is 
based on fast modulo-addition. For convenience, we state the algorithm in the Appendix A, where we 
also detail our implementation used in Section \ref{sec:experiments}.

When the number of clients is larger, we instead use the Distributed Compute Algorithm (DCA)
given in Algorithm~\ref{alg:DCA}, which we originally presented in a preliminary 
conference paper \cite{Heikkila_2017}. DCA is based on  
secret sharing using intermediate servers, called compute nodes,  
that do the actual summations. As before, real numbers are represented in fixed-point. 
The code on the compute server side could be run inside TEEs, if available, to increase 
security.

\begin{algorithm}[tb]
\caption{Distributed Compute Algorithm: secure summation for a large number of clients}
\label{alg:DCA}
\begin{algorithmic}[1]
\REQUIRE
   Number of parties $N$ (public); \\
   Number of compute nodes $M$ (public); \\
   Upper bound for the total sum $R$ (public); \\
   $y_i$ integer held by client $i, i=1, \dots, N$.
   \ENSURE Securely calculated sum $\sum_{i=1}^N y_i $.
   \STATE Each client $i$ simulates 
   $M-1$ vectors $u_{il}$ of uniformly random integers at most $R,$ and sets 
   $u_{iM} = -\sum_{l=1}^{M-1} u_{il} \mod R$.
   \STATE Client $i$ computes the messages
   $m_{i 1} = y_i + u_{i 1} \mod R$, $m_{i l} = u_{i l}, l=2,\dots M$,
   and sends them securely to the corresponding compute node.
   \STATE After receiving messages from all of the clients, 
   compute node $l$ broadcasts the noisy aggregate sums 
   $q_{l} = \sum_{i=1}^{N} m_{i l}$. A final aggregator will
   then add these to obtain $\sum_{l=1}^M q_{l} \mod R = \sum_{i=1}^N y_i$.
\end{algorithmic}
\end{algorithm}

In detail, we assume there are $M$ separate compute nodes, 
and use additive secret sharing for the summation, where the necessary number of shares 
for reconstructing the secret is $M.$ The security of the method relies on 
at least one of the nodes being honest-but-curious, i.e., not colluding with the others. 
As a setup, client $i$ runs a standard key-exchange 
with each of the compute nodes to establish shared keys for some symmetric encryption.

The shares $u_{il},l=1,\dots, M-1$ should be generated from the discrete uniform distribution bounded by $R$.

In the proposed method, the main setup cost is the key exchange protocol each client needs 
to run with the compute nodes. 
The main computational cost is typically incurred in generating the shares, and the communication 
scales linearly with the number of compute nodes.

To avoid possible bottlenecks on the compute nodes when $N$ is very large, we 
can increase $M$ while having each client only send messages to some smaller subset 
of the nodes. In this case at least one node in each subset needs to be non-colluding 
to maintain the privacy guarantees.


\subsection{Learning with minibatches}

For iterative learning, sub-sampling is often essential for good performance. This is 
especially true for DP learning with limited dataset sizes due to privacy amplification results: 
the uncertainty induced by using only a batch of the full data augments privacy, when the 
identities of the samples in a batch are kept secret \cite{Kasiviswanathan_2011, Beimel_2013, Abadi2016, Balle2018, Koskela2019}. 
We therefore want a method for choosing a batch [b] = $\cup_{i=1}^N [b_i]$ in the distributed setting s.t. the full batch 
is drawn according to a given sampling scheme and the identity of the batch participants is 
kept hidden.

Next, we formulate a general distributed sampling method that avoids the possibility of a catastrophic failure 
and guarantees a graceful degradation of privacy even in the presence of malicious adversaries. 
The main idea is to generate a list of random tokens such that 
there are as many tokens as there are observations in the full joint dataset, and only 
the owner of a given token knows who it belongs to.

Given such a list, we can use a simple joint sampling strategy: the participants jointly generate a 
seed for a PRNG, which is used to determine the batch participants without 
needing any additional communication. 
In case some batch $[b_i]$ is empty, the corresponding party sends only 
the noise necessary for DP. Note that 
this message cannot be distinguished from any other message, thus keeping 
the actual batch participants secret as required for privacy amplification.

To generate the list of tokens, we use Algorithm~\ref{alg:list_of_tokens} based on mixnets \cite{Chaum_1981}.

\begin{algorithm}[tb]
\caption{Create a list of tokens}
\label{alg:list_of_tokens}
\begin{algorithmic}[1]
\REQUIRE Each party $i$ has random integers $r_{ij},j=1,\dots,n_i$ called tokens that are all unique, and only $i$ 
knows who $r_{ij}$ belongs to.
There is a shared (arbitrary) ordering over the parties, assumed w.l.o.g. to be $1,\dots, N$. All parties know the public keys of all other parties.
	\STATE Party $i$ encrypts all it's tokens with the public keys of all parties according to the ordering: 
	$r_{ij} \rightarrow Enc_{k_1}( Enc_{k_2} (\dots Enc_{k_N}( r_{ij} )  \dots)), j=1,\dots,n_i $
	\STATE All parties publish all encrypted values to form a list with some ordering.
	\FOR{ $i=1$ to $N$ } \label{alg3:for_loop}
		\STATE Party $i$ decrypts layer $i$ from all the messages on the current list, randomly permutes the elements, and publishes the resulting new list. \label{alg3:loop_iter}
	\ENDFOR
\STATE Returns a list of tokens $r_{ij},i=1,\dots,N, j=1,\dots,n_i$ s.t. only the owner of a token knows who it belongs to.
\end{algorithmic}
\end{algorithm}

\begin{theorem}
\label{thm:list_of_tokens}
Assume the public-key encryption used in Algorithm \ref{alg:list_of_tokens} is secure. 
Then 
i) the list cannot be manipulated in nontrivial ways,  
ii) to hbc parties all tokens belonging to other parties are indistinguishable from each other and 
to malicious parties all tokens belonging to hbc parties are indistinguishable from each other.
\end{theorem}
\begin{proof}

See Appendix B.

\end{proof}

Generating the list of random tokens is an expensive step by itself in terms of communication and 
computations, but it suffices to create the list only once since the owners of the tokens are not revealed 
at any point when using the list. 
The privacy guarantees resulting from using the list depend on the actual 
sampling scheme used.


\subsubsection{Comparison of common sampling schemes}
\label{sec:sampling_comparison}

As noted, the two most used sampling schemes with privacy amplification are 
sampling without replacement (SWOR), where the minibatch size $b$ 
is fixed and each sample has equal probability $b/n$ of being selected, and Poisson sampling, 
where we fix the probability $\gamma \in (0,1)$ of independently including each observation 
while the batch size varies.

In the distributed setting, SWOR can be done using Algorithm \ref{alg:list_of_tokens}. 
However, Poisson sampling is also achievable in a more straightforward manner, 
since each client can simply 
choose to include each sample independently 
with the given probability. As shown in the rest of this Section, when 
malicious parties control some of the data and know if 
they are included in the batch or not, the simpler 
Poisson sampling leads to better privacy amplification than SWOR.

With SWOR, assume we sample a set of $b$ tokens uniformly 
at random from the set of all such sets from 
the list of tokens generated with Algorithm \ref{alg:list_of_tokens}, 
and let $\gamma = b/n$, so the mean batch 
size from Poisson sampling matches SWOR. 

Writing $[T]$ for the set of $T$ malicious parties, 
let the total number of samples controlled 
by the non-malicious parties be $n_{\neg T}=\sum_{i \not \in [T]} n_i$. 
Using SWOR, the actual number of samples coming from 
the non-malicious parties now follows a Hypergeometric 
distribution with total population size $n$, total number of successes in the 
population $n_{\neg T}$, 
and number of draws $b$.

A worst-case amplification factor is then given by $\frac{b}{n_{\neg T}}$
i.e., all the samples in a batch 
come from the non-malicious parties. However, in most settings this is a rather unlikely event 
since the distribution is concentrated around the mean $b \frac{n_{\neg T}}{n}$. 
We can therefore 
improve markedly on the worst-case if we admit some amount of slack in the privacy parameter $\delta$
as shown in Figure \ref{fig:SWOR_sampling_fracs}. 
However, the amplification factor is still   
worse than the baseline given by Poisson sampling when there are any 
malicious data holders present.

\begin{figure}[htb]%
\begin{center}%
	\centerline{\includegraphics[width=.8\columnwidth]{\figfolder/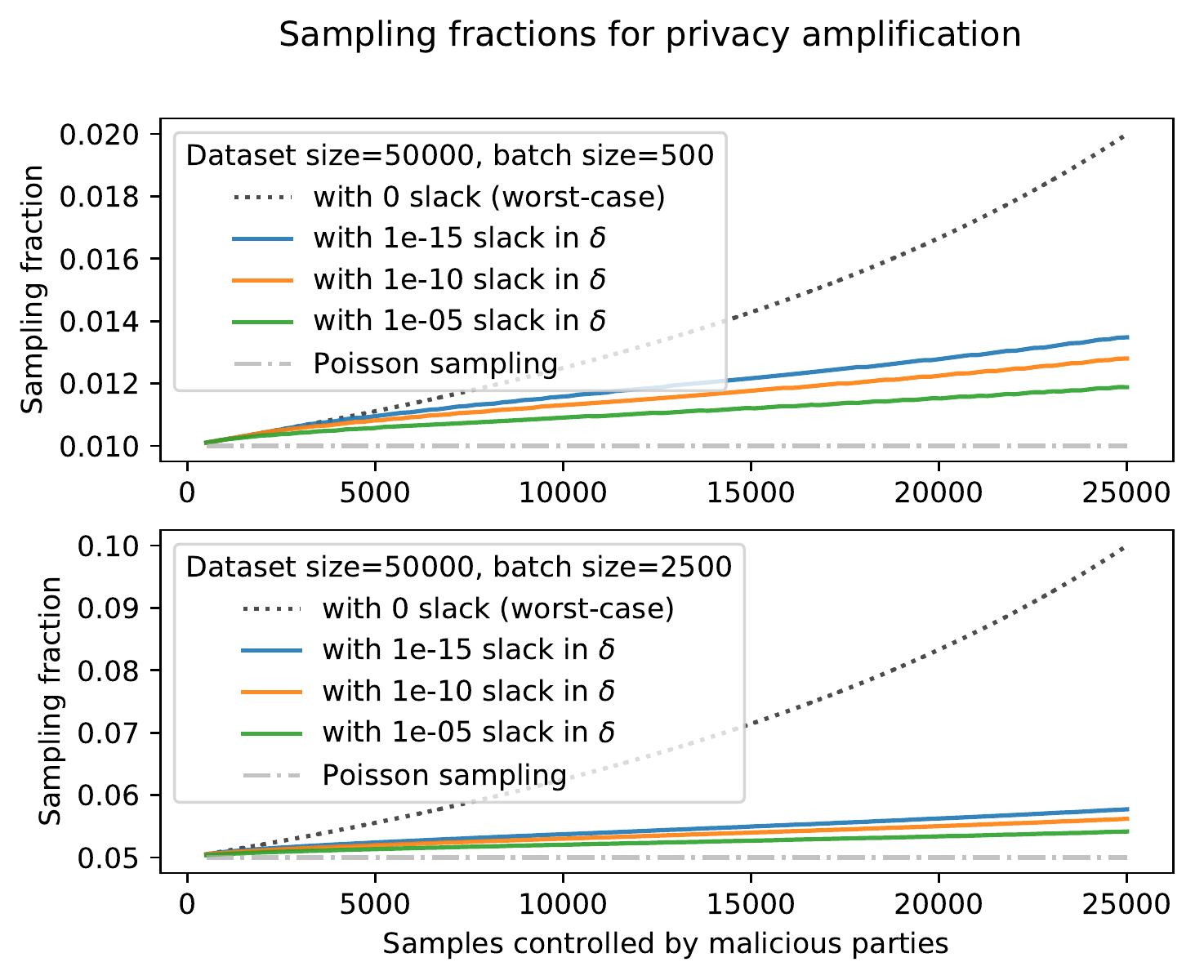}}%
	\caption{
	Effective sampling fraction for privacy amplification (lower is better). 
	With SWOR, allowing even for a tiny slack in $\delta$ improves markedly over the worst-case. The sampling fraction 
	for Poisson sampling is not affected by the malicious parties.
	\label{fig:SWOR_sampling_fracs}}%
	\end{center}%
\end{figure}%


\subsection{DP random projection}
\label{sec:random_projection}

To enable using the methods with even larger models, we propose using a low-dimensional random 
embedding instead of the full gradient in learning. Although the saving is only on uploaded 
gradients, this presents the most significant costs as it requires encryption, whereas downloading the 
updated parameter values can be done in the clear due to DP. In the distributed setting, we do not 
need to communicate the full projection matrix but only the seed for a PRNG.

\begin{algorithm}[H]
\caption{DP random projection}
\label{alg:dp_rand_proj}
\begin{algorithmic}[1]
\REQUIRE Clipped gradient vectors $z_{ij} \in \mathbb R^{d}, i\in \{1,\dots,N\}, j\in \{1,\dots,n_i\}$ with sensitivity $C$,
		projection dimension $k$,  
		projection sensitivity bound $\tilde C$, Gaussian mechanism 
		$\mathcal G_f$ that is $(\epsilon,\delta)$-DP on a sensitivity $\tilde C$ query.
	\STATE Generate a projection matrix $P \in \mathbb R^{ d \times k }$ s.t. each element is 
		independently drawn from $\mathcal N(0,1/k)$.
	\STATE Sum the projected clipped gradients: $ \tilde z = \sum_i \sum_j P^T z_{ij} $.
\STATE Return a DP projection $\mathcal G_f (\tilde z) $.
\end{algorithmic}
\end{algorithm}

We will next show that when a DP mechanism operates on 
the $k$-dimensional representation and the projection matrix is public, 
as described in Algorithm \ref{alg:dp_rand_proj}, the result is still DP. 
The result holds for both neighbouring relations $\sim_R, \sim_S$ we consider.

\begin{theorem}
\label{thm:dp_rand_projection}
Algorithm \ref{alg:dp_rand_proj} is $(\epsilon, \delta + \delta')$-DP, with any $\tilde C > 0$ and $\delta' > 0$ s.t. 
$$ \mathbb P \left[  \Gamma ( K=\frac{k}{2}, \theta=\frac{2C^2}{k} ) \leq \tilde C^2 \right] \geq 1-\delta'  ,$$
where $k$ is the projection dimension, $C$ is the gradient sensitivity, and $\Gamma$ is 
the (shape \& scale parameterised) Gamma distribution.
\end{theorem}

\begin{proof}

See Appendix B for a proof.

\end{proof}

To determine the sensitivity bound $\tilde C$ in Theorem $\ref{thm:dp_rand_projection}$, we decide 
on the additional privacy parameter $\delta' > 0$ and search for the smallest value of $\tilde C$ s.t. 
the condition in Theorem \ref{thm:dp_rand_projection} holds. Then, as stated in the Theorem, 
when the Gaussian mechanism we use is $(\epsilon,\delta)$-DP on a sensitivity $\tilde C$ query, 
the final result is $(\epsilon,\delta+\delta')$-DP.

In practice, each client locally runs step $(1)$ in Algorithm $\ref{alg:dp_rand_proj}$ 
with a shared random seed to get the individual noisy projected gradients and add Gaussian noise 
that will sum up to the desired variance as specified in Section \ref{sec:distributed_noise}. 
The aggregation is then done using the secure 
summation as discussed in Section \ref{sec:secure_sum}. 
After receiving the final decrypted sum $\mathcal G_f (\tilde z)$, 
the master inverts the projection by calculating 
$P \mathcal G_f (\tilde z) = P P^T [ \sum_i \sum_j z_{ij} + \eta_{ij} ] \simeq \sum_i \sum_j z_{ij} + \eta$, 
where $\eta_{ij}$ are the noise terms that sum up to the chosen DP noise as in \eqref{eq:dp_sgd_noisy_sum},  
and takes an optimization step with the resulting $d$-dimensional DP approximate gradient.


\subsection{Use cases}
\label{sec:scenarios}

We focus on two distinct scenarios:
\begin{enumerate}
\item \emph{Fat clients}: $N$ is small, the $n_i$ are moderate, the clients have fair computation and communication capacities
\item \emph{Thin clients}: $N$ is large, $n_i$ are small, the clients have more limited computation and communication capacities
\end{enumerate}

The two main scenarios we consider are motivated by considering $n_i$, the number of 
samples per client, and $N$, the number of clients. 
In order to decrease the effects of the DP noise 
we need more data, i.e., we need to increase either $n_i$ or $N$. As the $n_i$ grow the benefits 
from doing distributed learning start to vanish.

In other words, the most interesting regime for distributed learning 
ranges from each client having a single sample to each having a moderate amount of data.
As for the number of parties, increasing $N$ 
naturally increases the total amount of communication needed 
while also limiting the set of practical secure protocols.

A typical example of fat clients 
is research institutions or enterprises that possess some amount of data 
but would benefit significantly from doing learning on a larger joint dataset.

A stereotypical case of thin clients 
is learning with IoT devices, where the devices we mainly consider would 
be e.g. lower-level IoT nodes, which aggregate data from several simple edge devices.


\section{Experiments}
\label{sec:experiments}

As noted earlier, the code for all the experiments is freely available \cite{mikko_heikkila_2020_3938639}.

We first test our proposed solution for distributed learning using sampling without replacement and 
substitution relation $\sim_S$ 
with CIFAR10 dataset \cite{CIFAR10} and a convolutional neural network (CNN): 
the convolutional part uses 2 convolutional layers (CL) with kernel size 3, stride 1 and 64 channels, 
each followed by max pooling with kernel size 3 and stride 2. After the CL, the model has 2 fully connected (FC) layers with 
384 hidden units each with ReLU and a softmax classification layer. 
The same basic model structure has been previously used e.g. in \cite{Abadi2016}. 
Similarly to \cite{Abadi2016}, we pretrain the CL weights using the CIFAR100 dataset assumed to be public and 
keep them fixed during the training on private data. However, unlike \cite{Abadi2016}, besides having a different 
sampling scheme and neighbourhood relation, we also initialise the FC layer 
weights randomly, and not based on pretraining.

We note that the $\sim 75 \%$ accuracy for the non-private model is not near the state of the art for the CIFAR10 
(e.g. \cite{Huang_2018} reach $\sim 99\%$); 
the model architecture was chosen to illustrate how distributed DP affects the accuracy in a commonly 
used and reasonably well-understood standard model. 
The actual training is done using the well-known DP-SGD \cite{Song_2013, Abadi2016} but 
using the privacy accountant of \cite{Koskela2019}.
For testing encryption speed, we use the MNIST dataset \cite{MNIST} and 2 FC layers with 536 units each. 
This model, denoted model $1m$, has $\simeq 10^6$ trainable parameters.

The tests were run on clusters with  
2.1GHz Xeon Gold 6230 CPUs. All the tests use 1 core to simulate a single party. 
We report the results using our Python implementation of the algorithms. 
Run time results for fat clients can be found in Appendix C.

We use Python's Secrets module 
to generate randomness 
in Algorithm \ref{alg:DCA}, which internally utilises 
urandom calls. 
Generating a single 32-bit random integer with the Secrets module 
takes about $2.15 * 10^{-6}$ seconds.

The values are sent by writing them on a shared disk without additional symmetric encryption. 
In real implementations all communications would be done over the available communication 
bandwidth using some standard symmetric encryption such as AES. 
We do not include this nor the possible TEE component in the testing 
but focus on the feasibility of learning based on Algorithms \ref{alg:DCA} 
and \ref{alg:dp_rand_proj}.

The baselines for accuracy are given by the centralised trusted aggregator setting, 
as well as local DP (LDP) \cite{Duchi_2013}, which means that each participant locally adds enough noise 
to protect its contributions without any encryption scheme. As shown in Figure \ref{fig:acc_fat}, there 
is a tradeoff between privacy and accuracy. 
However, combining the secure summation protocols with DP, as we propose, provides significantly 
better accuracy than is possible by each party independently protecting their privacy.

\begin{figure}[htb]
\begin{center}
	\centerline{\includegraphics[width=.8\columnwidth]{\figfolder/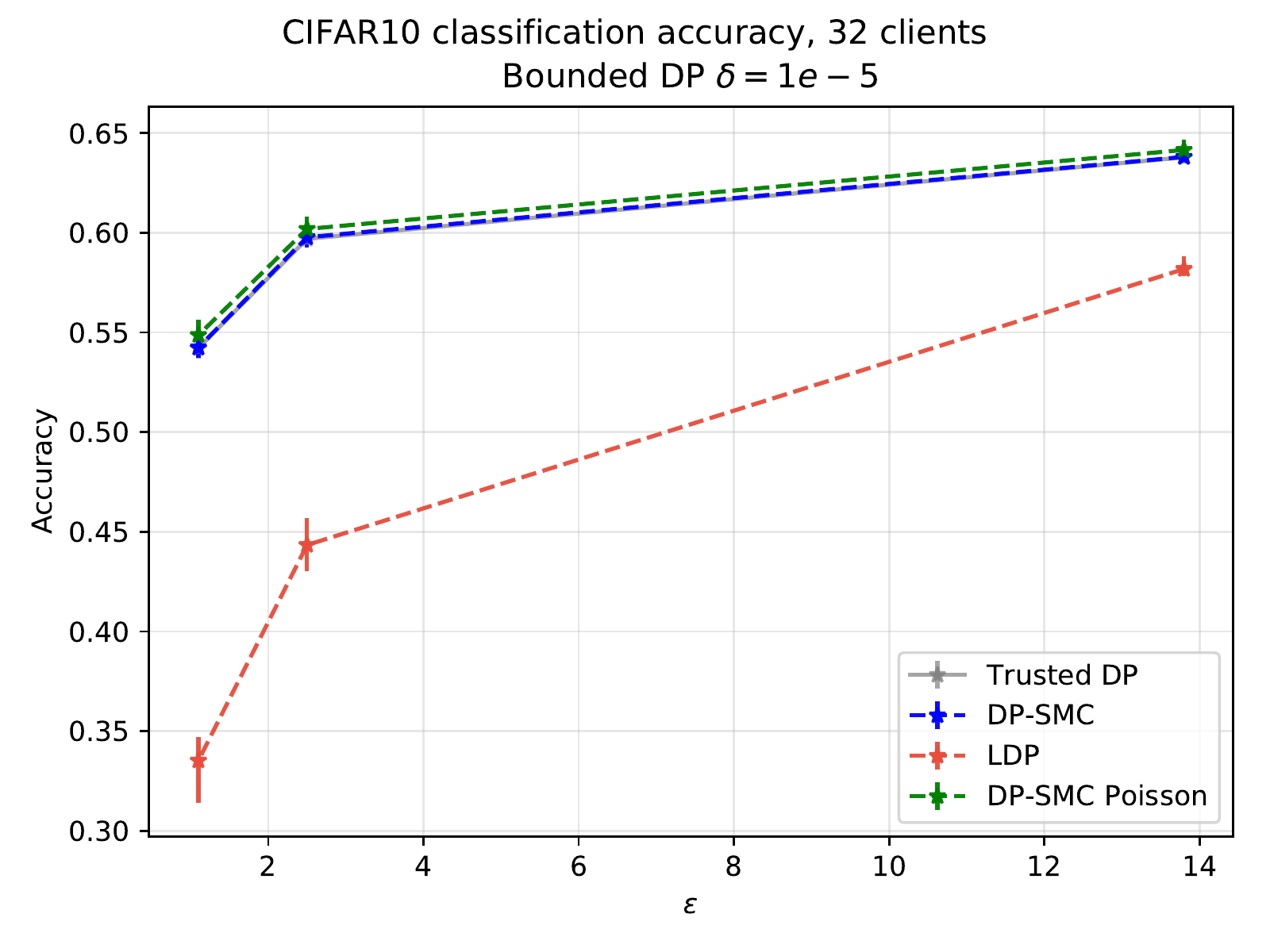}}
	\caption{Mean test accuracy and error bars showing min \& max over 5 runs with 
	fat clients on CIFAR10 data, $\delta=10^{-5}$. Our method (DP-SMC, DP-SMC Poisson) 
	gives significantly better results than  LDP especially in the stricter privacy regime $\epsilon < 2$. 
	The performance of DP-SMC using either sampling without replacement (DP-SMC) or Poisson sampling 
	(DP-SMC Poisson) is indistinguishable from the centralised scheme 
	with a trusted aggregator that uses either sampling without replacement (trusted DP) 
	or Poisson sampling (not shown).
	\label{fig:acc_fat}}
	\end{center}
\end{figure}

\begin{figure}[htb]%
\begin{center}%
	\centerline{\includegraphics[width=.8\columnwidth]{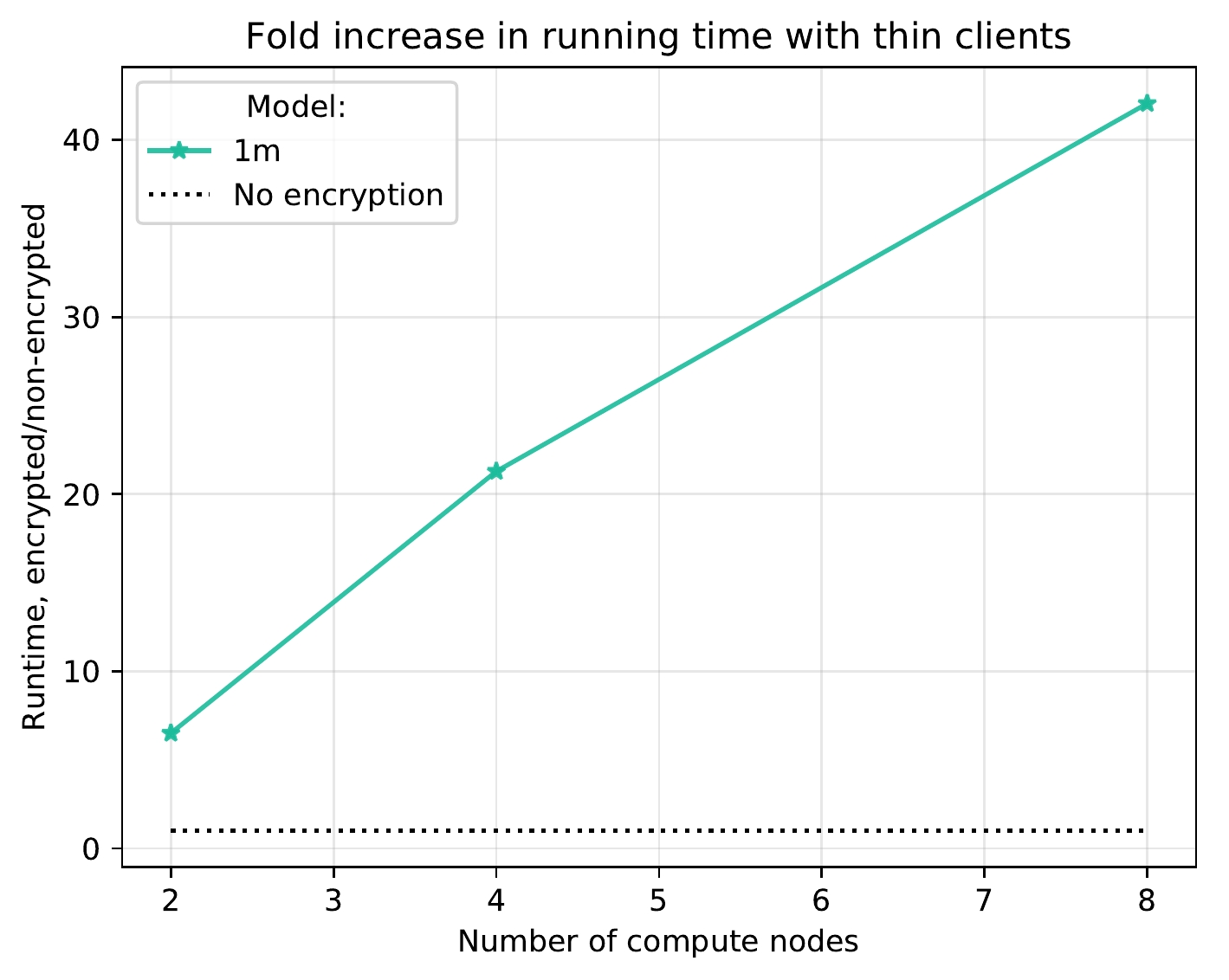}}%
	\caption{Fold increase in total running time with varying number of compute nodes for 100 clients, 
	medians over 5 runs. 
	The number of compute nodes can be tuned to adjust the tradeoff between security 
	and running time. 
	Using MNIST data, model $1m$ has 2 FC hidden layers with $\simeq 1e6$ parameters. 
	\label{fig:fold_increase_thin}}%
	\end{center}%
\end{figure}%

The tradeoff between running times and the security parameters is shown in 
Figure \ref{fig:fold_increase_thin}. 
The results show that our solutions can be realistically 
used for learning complex models with distributed data.

We can use Algorithm \ref{alg:dp_rand_proj} to reduce the computational burden due to 
cryptography and the amount of communication needed at the 
cost of having to generate projection matrices and losing some prediction accuracy. 
As shown in Figure \ref{fig:projection_acc}, even with very large savings on communication 
and decreased runtime, we still incur only modest decrease in accuracy.

\begin{figure}[htb]%
\begin{center}%
	\centerline{\includegraphics[width=.8\columnwidth]{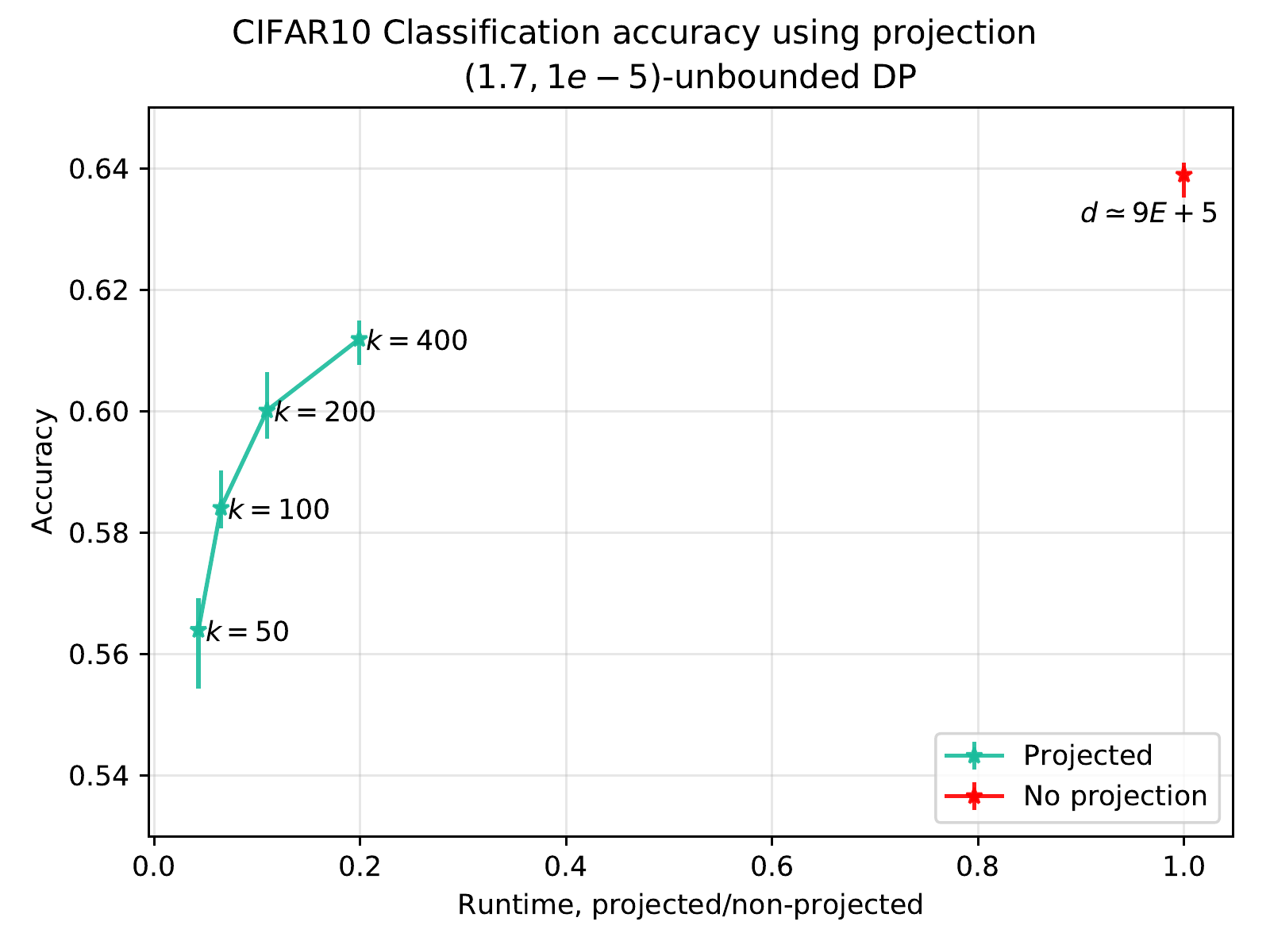}}%
	\caption{Mean test accuracy and error bars showing min \& max over 5 runs 
	using only low-dimensional embeddings vs. fold decrease in total running time, 
	Poisson sampling $(1.7,10^{-5})$-DP with 8 compute nodes. 
	The full model has $d \simeq 9\cdot 10^5$ parameters. 
	Accuracy is close to the original even with very low projection dimensionality $k$. 
	Low dimensionality has marked advance in total runtime and requires communicating 
	several orders of magnitude less encrypted values.
	\label{fig:projection_acc}}%
	\end{center}%
\end{figure}%


\section{Related work}
\label{sec:related_workj}

There is a wealth of papers on distributed DP with secure summation as witnessed by 
a recent survey \cite{Goryczka_2017}, many of which are relevant also for DP federated learning. 
The problem of distributed DP was first formulated in \cite{dwork_distributed_2006}, who 
propose to generate DP noise collaboratively using a SMC protocol. The idea of using a general SMC 
for distributed DP has been more recently considered e.g. in \cite{Jayaraman_2018}. The main 
drawback with these methods is the computational burden.

The first practical method for implementing DP queries in a distributed manner was the distributed
Laplace mechanism presented in \cite{Rastogi_2010}. Closely related methods have been 
commonly used in the literature ever since \cite{Shi_2011, Acs_2011, Goryczka_2017, Heikkila_2017, imtiaz2019distributed}, 
including in this work.

There has been much new work on distributed learning in the 
FL setting \cite{McMahan_2017}, as evidenced by a recent survey \cite{Kairouz_2019}.
Most of the work has concentrated on the cross device case with 
typically millions of clients each holding at most a few samples. 
Meanwhile, the cross-silo setting has seen significantly fewer contributions, especially 
focusing on privacy.

Of the closest existing work, 
in \cite{choudhury2020anonymizing} the authors consider privacy in cross-silo FL setting 
without DP, but their approach lacks any formal privacy guarantees. 
In \cite{Wei_2020} the focus is on deriving results for training models on the clients' 
local data under DP, where the privacy guarantees are derived directly on the communicated 
parameters instead of gradients. However, all the reported experimental results 
use privacy budgets that are completely vacuous in practice ($\epsilon \geq 50$). 
\cite{Truex_2019hybrid} concentrates on training several ML models in the cross-silo FL setting with DP. 
Their solution uses essentially the same distributed noise generation but 
slower encryption, does not use TEEs nor address sampling schemes other than the distributed 
Poisson sampling. They also do not consider communicating low-dimensional embeddings.

TEEs have recently been gaining popularity in machine learning literature. 
Privacy-preserving machine learning using TEEs was proposed in \cite{Ohrimenko_2016}, 
who discuss security of enclaves especially against data access pattern leaks. 
\cite{costa2017pyramid} concentrate on memory access leaks in SGX, whereas 
\cite{Allen_2018} propose a new DP definition to account for the possible privacy breaches resulting 
from data and memory access leaks. \cite{Hynes_2018} focus on how to implement efficient operations 
needed for a ML framework built on TEEs and DP. 
Several recent papers have investigated efficient execution of DNNs and other models in 
TEEs~\cite{tramer2018slalom, hunt2018chiron, gu2018yerbabuena, kunkel2019tensorscone}.


\section{Conclusion}

We have presented two methods for cross-silo FL with strict privacy-guarantees 
based on fast homomorphic encryption and distributed noise addition, and 
shown empirically that they can be used for learning complex models like 
neural networks with reasonable running times. Finally, 
we have shown how the scale the methods to even larger models by 
using random projections while incurring only modest cost in terms 
of prediction accuracy.



\subsubsection*{Acknowledgements}

The authors wish to acknowledge CSC -- IT Center for Science, 
Finland, for computational resources.

\bibliography{DP-cross-silo-FL-arXiv}


\section*{Appendix A}

This Appendix contains the secure summation protocol and related discussion for fat clients.

\subsection*{Secure summation with fat clients}

For secure summation with fat clients, 
we use an additively homomorphic encryption scheme given 
in Algorithm~\ref{alg:Castelluccia_sum}. 

The algorithm was originally introduced in 
\cite{Castelluccia_2005}, uses fixed-point representation of real numbers, and is 
based on fast modulo-addition.

\begin{algorithm}[H]
\caption{Secure summation for fat clients }
\label{alg:Castelluccia_sum}
\begin{algorithmic}[1]
\REQUIRE
   Upper bound for the total sum $R$ (public); \\
   $y_i$ integer held by party $i, i=1, \dots, N$; \\
   Pairwise secret keys $k_{i j}$ held by party $i, \ i,j=1,\dots, N, i \neq j$ s.t. $k_{i j} + k_{j i} = 0 \mod R$ .
   \ENSURE Securely calculated sum $\sum_{i=1}^N y_i $.
   \STATE Each client $i$ calculates $Enc(y_{i}, k_{i}, R ) = y_{i } + k_{i} =  y_{i }+ \sum_{j\neq i} k_{ij} \mod R,$ 
   and sends the result to the aggregator.
   \STATE After receiving messages from all other parties, 
   the aggregator broadcasts the sum $\sum_{i=1}^N \sum_{j\neq i} y_{i } + k_{i j} \mod R = \sum_{i=1}^N y_i$.
\end{algorithmic}
\end{algorithm}

Here we assume the parties are TEEs but this is not essential for the protocol. 
We assume one of the TEEs acts as an untrusted aggregator, who does the actual summations 
and broadcasts the results to the other TEEs but otherwise has no special information or capabilities. 
Without implementing e.g. some zero-knowledge proof of validity, the aggregator might change 
the final result at will if the TEE is compromised. However, this can usually 
be detected post hoc by each TEE by comparing 
the jointly trained model to the model trained on their private data. The role of the master could also 
be randomised or duplicated if deemed necessary.

For the secure DP summation, the TEEs encrypt $y_i = z_{i} + \eta_{i}$ where $\eta_i$ is DP 
noise as in the main text's Section \ref{sec:distributed_noise}\footnote{
Note that the actual sum is unbounded with a Gaussian noise term for DP. However, 
we can always clip the values to the assumed range since post-processing does not affect the DP guarantees.
}.

To generate the secret keys $k_{ij},$ for each round a separate setup phase is needed, which can be done with $O(N
^2)$ messages e.g. using standard Diffie-Hellman key-exchange \cite{Diffie_1976} or existing public key cryptography. 
The actual values for each iteration can be drawn from the discrete uniform distribution bounded by $R$ 
using a cryptographically secure pseudorandom 
generator (CSPRNG) \cite{Goryczka_2017} initialised with a seed shared by each pair of TEEs. Since each individual's value 
is protected by all pairwise keys, the method is secure as long as there are at most $N-2$ compromised 
TEEs. We refer to \cite{Castelluccia_2005} and \cite{Acs_2011} 
for more details.


\section*{Appendix B}

This Appendix contains all the proofs omitted from the main text. For convenience, we 
first state the theorems again and then proceed with the proofs.

\subsection*{Proof of Theorem 1}

\setcounter{theorem}{0}

\begin{theorem}
Assume the public-key encryption used in Algorithm \ref{alg:list_of_tokens} is secure. 
Then 
i) the list cannot be manipulated in non-trivial ways,  
ii) to hbc parties all tokens belonging to other parties are indistinguishable from each other and 
to malicious parties all tokens belonging to hbc parties are indistinguishable from each other.
\end{theorem}
\begin{proof}

To begin with, since the encryption is secure, an attacker cannot deduce the true ciphertext at layer $i$ 
by seeing the unencrypted plaintext at layer $i+1, i=1,\dots, N-1$.

For i), since after each iteration \eqref{alg3:loop_iter} in the for-loop any party can check if its tokens are still included on the list and that the list size has not changed, 
the only manipulation that can be done without alerting some hbc party is to tamper with the tokens of the malicious 
parties without changing their total number. Since these parties can in any case decide their contribution to learning without 
regarding the tokens, manipulating these tokens cannot gain anything.

As for ii), w.l.o.g. assume hbc parties at iterations $i$ and $j$, $i < j$ in the for-loop \eqref{alg3:for_loop}. 
For $j$, since $i$ follows the protocol the resulting list after iteration $i$ of \eqref{alg3:loop_iter} looks like random numbers with a random permutation, 
so all elements not belonging to $j$ are indistinguishable from each other. 
For $i$, even if $i$ can identify all elements on the list after iteration $i$ of the for-loop, 
after iteration $j$ the list again looks like a random permutation of random numbers since $j$ follows the protocol. Furthermore,
after iteration $i$ the entries belonging to the hbc parties look like random numbers with a random permutation 
to any malicious parties and as such are indistinguishable from each other to any malicious party from this iteration onwards.

\end{proof}

\subsection*{Proof of Theorem 2}

\begin{theorem}
Algorithm \ref{alg:dp_rand_proj} is $(\epsilon, \delta + \delta')$-DP, with any $\tilde C > 0$ and $\delta' > 0$ s.t. 
$$ \mathbb P \left[  \Gamma ( K=\frac{k}{2}, \theta=\frac{2C^2}{k} ) \leq \tilde C^2 \right] \geq 1-\delta'  ,$$
where $k$ is the projection dimension, $C$ is the gradient sensitivity, and $\Gamma$ is 
the (shape \& scale parameterised) Gamma distribution.
\end{theorem}

\begin{proof}

Let $f_{JL}(a) = P^T a $ be the projection with each element in $P \in \mathbb R^{d \times k }$ 
independently drawn from $\mathcal N(0,1/k)$. 
We need to show that $f_{JL}$ is 
$(\tilde C, \delta')$-sensitive (Definition \ref{def:function_sensitivity}), and 
the result then follows directly from Lemma \ref{lemma:dp_projection}.

Consider first the neighbouring relation $\sim_R$. The sums of clipped gradients 
originating from maximally different data sets $D~\sim_R~D'$ differ in one vector, w.l.o.g. 
denoted as $a \in \mathbb R^d, \| a \|_2 \leq C$. 
As the coupling required by Definition \ref{def:function_sensitivity}, we use 
a trivial independent coupling, so we have $X - X' = f_{JL}( a)= P^T a$.

Writing $\mathcal N, \chi^2$ for random variables following normal and chi-squared distributions, 
respectively, we therefore have 
\begin{align}
\| X - X' \|^2_2 &=  \| P^T a  \|^2_2 \\
&=  \sum_{j=1}^k [\sum_{i=1}^d  a_i \mathcal N(0,1/k)  ]^2 \\
&= \frac{ \| a \|_2^2 }{k}  \chi_k^2 \\
&\leq \frac{ C^2 }{k}  \chi_k^2 .
\end{align}

Looking at Definition \ref{def:function_sensitivity}, we see that $f_{JL}$ is 
$(\tilde C, \delta')$-sensitive, when 
\begin{equation}
\label{eq:proj:condition}
\mathbb P \left[  \Gamma ( K=\frac{k}{2}, \theta=\frac{2C^2}{k} ) \leq \tilde C^2 \right] \geq 1-\delta'  .
\end{equation}

The fact that Algorithm \ref{alg:dp_rand_proj} is $(\epsilon, \delta+\delta')$-DP 
then follows from Lemma \ref{lemma:dp_projection}.

When using the neighbouring relation $\sim_S$ and constant $C/2$ for clipping the 
gradients, the gradients originating from the maximally different data 
sets $D~\sim_S~D'$ again differ at most in a vector $a$ s.t. $\| a \|_2 \leq C$. 
The privacy analysis of the mechanism $\mathcal G_f $ is then carried out using 
the techniques given in \cite{Koskela2019} 
and we get the $(\epsilon,\delta)$-bound for $\mathcal G_f$ using the projection 
sensitivity bound $\tilde C$, so again Algorithm \ref{alg:dp_rand_proj} is 
$(\epsilon, \delta+\delta')$-DP for $\delta'$ that satisfies \ref{eq:proj:condition}.

\end{proof}


\section*{Appendix C}

This Appendix contains the experimental results for fat clients, i.e., 
using Algorithm \ref{alg:Castelluccia_sum} for encryption, omitted 
from the main text.

\subsection*{Implementation details and more experimental results}

To generate randomness with a CSPRNG required in Algorithm \ref{alg:Castelluccia_sum}, we use Blake2 \cite{Blake2}, 
a fast cryptographic hash function, on 
an initial pairwise shared secret together with a running number.
On 2.1GHz Xeon Gold 6230, generating one 4 byte hash digest with Blake2 takes 
about $1.2 * 10^{-6}$ seconds.

The tradeoff between running times and the security parameters corresponding to 
the main text's Figure \ref{fig:fold_increase_thin} is shown in 
Figure \ref{fig:fold_increase_fat}.

\begin{figure}[htb]%
\begin{center}%
	\centerline{\includegraphics[width=.8\columnwidth]{\figfolder/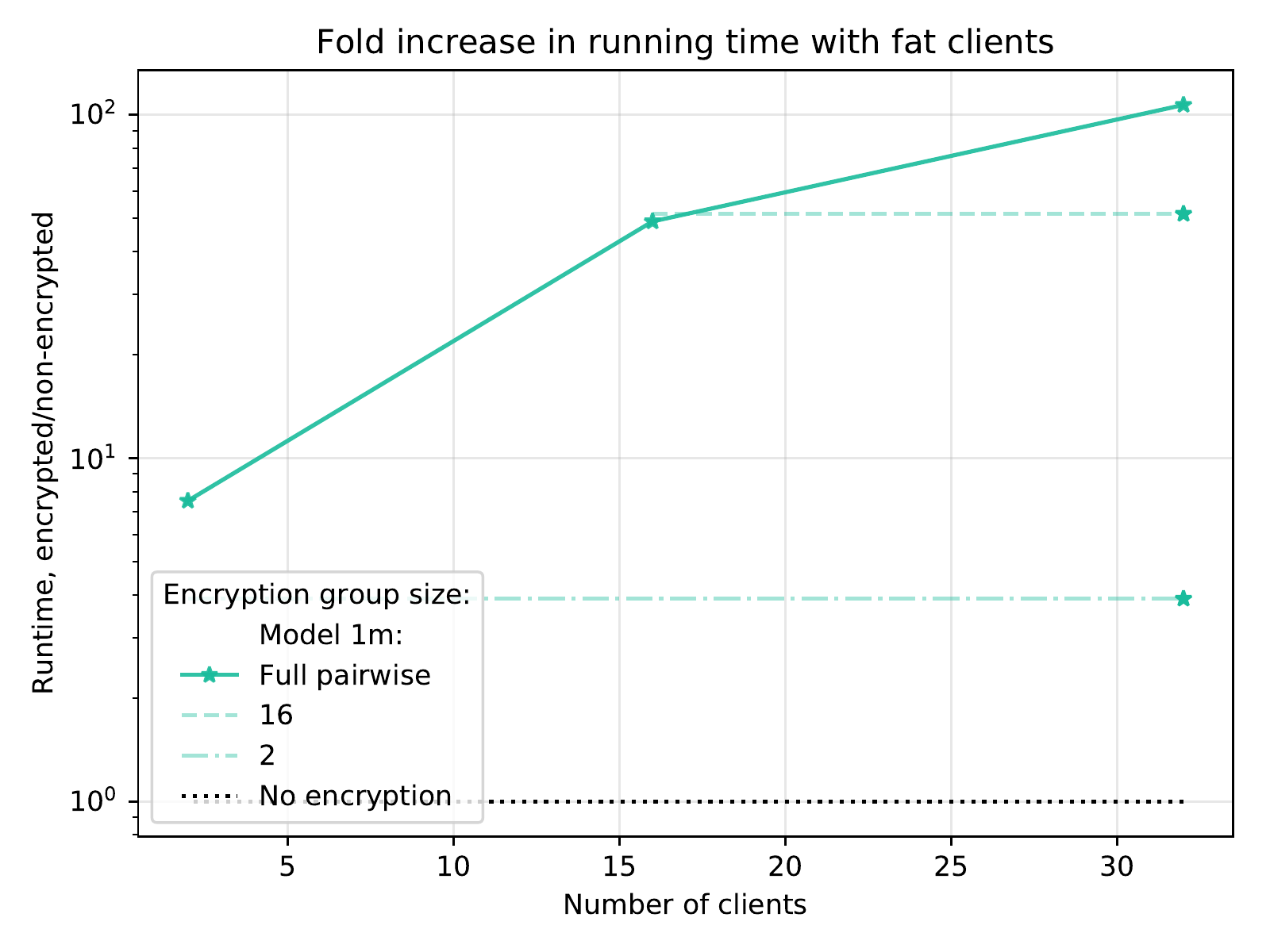}}%
	\caption{Fold increase in total running time with varying number of clients, medians over 5 runs. 
	Encryption time depends on the chosen pairwise group size: even with increasing number 
	of clients the encryption time stays roughly constant with a fixed group size (horizontal lines). 
	Using MNIST data, model $1m$ has 2 FC hidden layers with $\simeq 1e6$ parameters. 
	\label{fig:fold_increase_fat}}%
	\end{center}%
\end{figure}%


\end{document}